\pdfoutput=1 
\documentclass[11pt]{article}
\usepackage[utf8]{inputenc}
\usepackage[T1]{fontenc}
\usepackage[sc]{mathpazo}
\usepackage{microtype}
\usepackage{amsmath,amssymb,amsthm}
\usepackage{authblk}
\usepackage{fullpage}
\usepackage[usenames,dvipsnames]{color}
\usepackage{hyperref}
\usepackage{xspace}
\usepackage{mathtools}
\usepackage{enumitem}
\usepackage{multicol}
\usepackage[scaled]{helvet}
\usepackage{thmtools}
\usepackage{cleveref}
\usepackage{comment}
\usepackage[ruled,vlined,linesnumbered]{algorithm2e}
\usepackage{xspace}
\usepackage[labelfont=bf]{caption}

\definecolor{blue}{RGB}{0,50,200}
\definecolor{green}{RGB}{92,183,51}
\hypersetup{colorlinks,linkcolor=blue,urlcolor=blue,citecolor=magenta}

\title{An Improved Algorithm For Online Min-Sum Set Cover}

\author[1]{Marcin Bienkowski}
\author[2]{Marcin Mucha}
\affil[1]{Institute of Computer Science, University of Wrocław, Poland}
\affil[2]{Institute of Informatics, University of Warsaw, Poland}
\date{}

\newcommand{\shortcite}[1]{\cite{#1}}

\newtheorem{theorem}{Theorem}
\newtheorem{lemma}[theorem]{Lemma}
\newtheorem{corollary}[theorem]{Corollary}

\newtheorem{observation}[theorem]{Observation}
\newtheorem{fact}[theorem]{Fact}

\newcommand{\OPT}{\normalfont{\textsc{Opt}}\xspace}
\newcommand{\ALG}{\normalfont{\textsc{Alg}}\xspace}
\newcommand{\LM}{\normalfont{\textsc{Lma}}\xspace}
\newcommand{\OPTH}{\ensuremath{\OPT^\mathrm{E}}\xspace}
\newcommand{\OPTS}{\ensuremath{\OPT^\mathrm{S}}\xspace}
\newcommand{\ALGH}{\ensuremath{\ALG^\mathrm{E}}\xspace}
\newcommand{\ALGS}{\ensuremath{\ALG^\mathrm{S}}\xspace}

\newcommand{\OFFH}{\ensuremath{\textsc{Off}^\mathrm{E}}\xspace}
\newcommand{\MTFS}{\ensuremath{\textsc{Mtf}^\mathrm{S}}\xspace}
\newcommand{\piMTF}{\ensuremath{\pi^\mathrm{MTF}}}
\newcommand{\I}{\mathcal{I}}
\newcommand{\IS}{\ensuremath{\mathcal{I}^\mathrm{S}}}
\newcommand{\JS}{\ensuremath{\mathcal{J}^\mathrm{S}}}
\newcommand{\IH}{\ensuremath{\mathcal{I}^\mathrm{E}}}
\newcommand{\swap}{\normalfont{\textsc{swap}}}
\newcommand{\fetch}{\normalfont{\textsc{fetch}}}
\newcommand{\size}{\normalfont{\textsc{size}}}
\newcommand{\pt}{\normalfont{\textsc{cp}}}
\newcommand{\E}{\mathbf{E}}
\newcommand{\fa}{\ensuremath{\alpha}}
\newcommand{\fb}{\ensuremath{\beta}}
\newcommand{\fc}{\ensuremath{\gamma}}
\newcommand{\fk}{\ensuremath{\kappa}}

\begin{document}

\maketitle

\begin{abstract}
We study a fundamental model of online preference aggregation, where an
algorithm maintains an ordered list of $n$ elements. An input is a stream of
\emph{preferred sets} $R_1, R_2, \dots, R_t, \dots$. Upon seeing $R_t$ and
without knowledge of any future sets, an algorithm has to \emph{rerank elements}
(change the list ordering), so that at least one element of $R_t$ is found near
the list front. The incurred cost is a sum of the list update costs (the number
of swaps of neighboring list elements) and access cost (the position of the
first element of $R_t$ on the list). This scenario occurs naturally in
applications such as ordering items in an online shop using aggregated
preferences of shop customers. The theoretical underpinning of this problem is
known as Min-Sum Set~Cover.

Unlike previous work (Fotakis et al., ICALP 2020, NIPS 2020) that mostly studied
the performance of an online algorithm \ALG in comparison to the \emph{static}
optimal solution (a~single optimal list ordering), in this paper, we study an
arguably harder variant where the benchmark is the provably stronger optimal
\emph{dynamic} solution \OPT (that may also modify the list ordering). In terms
of an online shop, this means that the aggregated preferences of its user base
evolve with time. We construct a computationally efficient randomized algorithm
whose competitive ratio (\ALG-to-\OPT cost ratio) is $O(r^2)$ and prove the
existence of a deterministic $O(r^4)$-competitive algorithm. Here, $r$ is the
maximum cardinality of sets $R_t$. This is the first algorithm whose ratio does
not depend on $n$: the previously best algorithm for this problem was $O(r^{3/2}
\cdot \sqrt{n})$-competitive and $\Omega(r)$ is a lower bound on the performance
of any deterministic online algorithm.
\end{abstract}


\section{Introduction}

We focus on the problem of maintaining an ordered (ranked) list of elements and
updating the order to better reflect the preferences of users. This problem
occurs naturally in an~online shop that has to present a list of items in some
order to users. Items that a user is interested in should be placed sufficiently
close to the list beginning: otherwise, a user has to scroll down, which could
degrade the overall experience and reduce customer retention. Similar phenomena
occur not only in online shopping~\cite{DeGoMM20}, but also in ordering results
from a web search for a given keyword~\cite{DwKuNS01,AgBrDu18}, or ordering news
and advertisements.

In the Min-Sum Set Cover (MSSC) problem, which serves as a theoretical model for
this problem, there is a universe $U$ of $n$ elements, and a set of $m$ users,
where the $t$-th user has a set of preferred elements $R_t \subseteq U$. The
goal is to find a fixed permutation of $n$ elements, which minimizes the sum of
users' dissatisfaction, where the dissatisfaction of user $t$ is the position of
the first element from $R_t$ in the permutation. This measures (in a perhaps
simplistic way) how far a user has to scroll till an interesting item is found.
This problem and its variants have been thoroughly studied in the approximation
algorithms community: the best polynomial-time solution is a
$4$-approximation~\cite{BaBHST98} and this approximation factor is tight unless
$\textsf{P} = \textsf{NP}$~\cite{FeLoTe04}.

In this paper, we study \emph{an online variant} of the MSSC problem, where an
algorithm has to maintain a permutation of $n$ elements, and preferred sets
appear in an online manner. Upon seeing a set $R_t$, an algorithm (i) first has
to pay an access cost equal to the position of the first element from $R_t$;
(ii) may reorder the list arbitrarily, paying the Kendall tau distance between
old and new permutation (minimal number of swapped adjacent
elements).\footnote{A careful reader may observe that the same cost measure is
applied both to accessing the first element and to swapping two adjacent list
elements. This choice is made to make the model coherent with the previous
papers. However, one can easily set the swapping cost to an arbitrary constant:
by the expense of small constant factors, these variants are reducible to our
problem using standard rent-or-buy approaches~\cite{KaMaMO94}.} This setting
captures a~frequent case where the service (e.g., a shop) is learning user
preferences on the fly and has to react accordingly, without knowing the
preferences of future users.


\subsection{Competitive ratio}

To measure the effectiveness of online algorithms we use one of the standard
yardsticks, namely competitive analysis~\cite{BorEl-98}, and we compare the
overall cost of an~online algorithm \ALG to an~optimal (offline) solution \OPT
on the same input instance $\I$. We emphasize that we compare our algorithm to
an optimal solution that \emph{can also change the permutation dynamically}. 

An algorithm \ALG is $c$-competitive if there exists $\xi$ such that for any
input instance $\I$, it holds that $\ALG(\I) \leq c \cdot \OPT(\I) + \xi$. The
competitive ratio of $\ALG$ is the infimum of values of $c$, for which \ALG is
$c$-competitive. For randomized algorithms, we replace the cost $\ALG(\I)$ with
its expected value $\E[\ALG(\I)]$, where the expectation is taken over random
choices of \ALG.


\subsection{Previous and our results}

Fotakis et al.~\shortcite{FoKKSV20} studied the online MSSC problem and
constructed an online $O(r^{3/2} \cdot \sqrt{n})$-competitive algorithm
\textsc{Move-All-Equally} (\textsc{Mae}), where $r$ is the maximum cardinality
of requested sets $R_t$. Their solution, for the requested set $R_t$, computes
the position $d$ of the first element from $R_t$ in the current permutation and
moves all elements $d-1$ positions towards the list front. They showed that the
competitive ratio of any deterministic algorithm cannot be lower than
$\Omega(r)$ and proved that many natural algorithms based on the
\emph{move-to-front} heuristic have a competitive ratio of $\Omega(n)$. They
also asked whether attaining a competitive ratio which is only a~function of $r$
is possible. (Note that $r \ll n$ for most practical applications).

We answer this question affirmatively, providing a novel approach, which allows
us to construct a randomized $O(r^2)$-competitive algorithm. Moreover, our
result holds even if each set $R_t$ is chosen by an adversary based on the
current state of the algorithm's list, i.e., holds also against
\emph{adaptive-online} adversaries~\cite{BorEl-98}. This implies the
\emph{existence} of a deterministic $O(r^4)$-competitive
algorithm~\cite{BeBKTW94}. We note that using the techniques of Ben-David et
al.~\shortcite{BeBKTW94}, it is possible to construct such a deterministic
algorithm. The comparison of old and new results is given in Table~\ref{tab:results}.

On the technical level, to solve the problem, in Section~\ref{sec:caching} we
introduce the Exponential Caching~(EC) problem and show that the MSSC problem
reduces to EC with the loss of constant factors. Glossing over details, EC
treats the list as being split into chunks of geometrically growing sizes and
captures the intuition that --- neglecting constant factors --- the costs depend
only on the chunk index. In Section~\ref{sec:algorithm}, we provide an algorithm
for the EC problem. To circumvent the lower bound for the algorithm MAE, we move
only the element that is closest to the list front, and we increase the budgets
of the remaining elements in the requested set instead. Once the budgets become
sufficient to pay for the element movement, the respective elements are moved to
the first chunk.


\subsection{Related work: static optimality}

The online variant of the MSSC problem has been also considered in an easier
setting where an online algorithm is compared to a \emph{static} optimal
solution that has to stick to a single permutation for the whole
runtime~\cite{FoKKSV20}. We emphasize that this variant differs from the
\emph{online learning} setting; that is, we assume that an online algorithm is
still charged for changing its permutation.

The static model forfeits optimization possibilities that occur when the
preferences of the user base are evolving (e.g., due to influences from
advertisements or because of seasonality). It is also worth mentioning that the
costs of static and dynamic optimal solutions can differ by a factor
of~$\Omega(n)$~\cite{FoKKSV20}. 

The randomized $O(1)$-competitive solution follows by combining multiplicative
weight updates (\textsc{Mwu})~\cite{LitWar94,ArHaKa12} with the techniques of
Blum and Burch~\shortcite{BluBur00} designed for the metrical task systems. This
approach has been derandomized by Fotakis et al.~\shortcite{FoKKSV20}, who gave
a~deterministic solution with an asymptotically optimal ratio of $\Theta(r)$.

\begin{table}
\centering
\begin{tabular}{ r|c|c|c } 
  &  & dynamic $\OPT$ & static $\OPT$ \\ 
\hline
\LM & det & $O(r^4)$ & $^\ast\,O(r^4)$ \\ 
\LM & rand & $O(r^2)$ & $^\ast\,O(r^2)$  \\ 
\hline
 \textsc{Mae} & det & $O(r^{3/2} \cdot \sqrt{n})$ & $2^{O(\sqrt{\log n \cdot \log r}})$ \\
 \textsc{Mwu} derand. & det & ? & $O(r)$ \\ 
 Lower bound & det & $^\ast\,\Omega(r)$ & $\Omega(r)$ \\
 \textsc{Mwu} & rand & ? & $O(1)$ \\ 
\hline
\end{tabular}
\caption{Competitive ratios of old algorithms and algorithm \LM presented in this paper,
against dynamic and static \OPT. Asterisked entries are trivially implied by other ones.}
\label{tab:results}
\end{table}


\subsection{Other related work}

The variant of the problem where all sets $R_t$ are singletons, known as the list
update problem, has been studied in a~long line of work and admits $O(1)$-competitive 
solutions, see~\cite{Kamali16} and references therein. 

Another line of work studied a generalization of the MSSC problem where each set
$R_t$ comes with a~covering requirement $k_t$ and an algorithm is charged for
the positions of the first $k_t$ elements from $R_t$ on the list. (The original
MSSC corresponds to $k_t = 1$ for any $t$). Known solutions include
$O(1)$-approximation (offline)
algorithms~\cite{AzGaYi09,BaGuKr10,SkuWil11,ImSvZw14,BaBaFT21} and
$O(1)$-competitive polynomial-time solution against static optimum without
reordering costs~\cite{FoLiPS20}.

Finally, a large amount of research was devoted to efficiently learning a
permutation with limited feedback, see,
e.g.,~\cite{HelWar09,YaHKTT11,YaHaTT12,SlRaGo13,Ailon14}. While the general aim
is similar to ours, the specific objectives and cost measures make these results
incomparable to ours.


\subsection{Problem definition and notation}

In the Min-Sum Set Cover (MSSC) problem, we are given a universe $U$ of $n$
elements. For the sake of notation, we assume that a permutation of $U$ is given
as a bijective mapping $U \to \{1, \dots, n\}$, returning for any element $x \in
U$ its position on the ordered list. 

An input $\I$ consists of an initial permutation $\pi_0$ of $U$ and a sequence
of $m$ sets $R_1, R_2, \dots, R_m$. Upon seeing set $R_t$, an online algorithm
\ALG is first charged the \emph{access cost} $\min_{x\in R_t} \pi(x)$. Then \ALG
chooses a new permutation $\pi_t$ (possibly $\pi_t = \pi_{t-1}$) paying
\emph{reordering cost} $d(\pi_{t-1}, \pi_{t})$, defined as the number of
inversions between $\pi_{t-1}$ and $\pi_t$. Note that $d(\pi_{t-1}, \pi_t)$ is
also the minimum number of swaps of adjacent elements necessary to change
permutation $\pi_{t-1}$ into $\pi_t$. We emphasize that the choice of~$\pi_t$
made by \ALG has to be performed without the knowledge of future sets $R_{t+1},
R_{t+2}, \dots$ and also without the knowledge of the sequence length $m$. 

In the following, \OPT denotes an optimal offline algorithm. For an input $\I$
and an algorithm $A$, we use $A(\I)$ to denote the total cost of $A$ on $\I$,
and $\ALG(\I,t)$ to denote the cost of $A$ in response to set $R_t$. 
For an integer $j$, we use $[j] = \{0, \ldots, j-1\}$.


\section{Exponential Caching Problem}
\label{sec:caching}

Without loss of generality,  in the MSSC problem, we may assume that the
universe cardinality is $n = 2^w - 1$, where $w \geq 1$ is an integer. To see
this, observe that it is always possible to add dummy elements that are never in
any requested set so that $n$ is of this form; these dummy elements are kept
by \OPT at its list end, and thus they do not increase its cost. After such
modification, the number of elements remains asymptotically the same.


\subsection{Problem definition}
\label{sec:ec_definition}

We now define an Exponential Caching (EC) problem, whose solution will imply the
solution to the MSSC problem of the asymptotically same ratio. In the EC
problem, an algorithm has to maintain a~time-varying partition of elements into
$w$ sets, henceforth called \emph{chunks}, whose sizes are powers of $2$. That
is, an~algorithm has to maintain a partitioning $p : U \to [w]$. We call $p(x)$
the \emph{chunk index} of element $x$ and we say that partitioning $p$ is
\emph{valid} if $|p^{-1}(i)| = 2^i$.  

We define chunks $S^p_0, S^p_1, \dots, S^p_{w-1}$, where $S^p_i = p^{-1}(i)$. We
usually skip $p$ in superscript if it does not lead to ambiguity. For any valid
partitioning $p$ and element $x$, we use
\[
  \size(p,x) = 2^{p(x)}
\]
to denote the cardinality of the chunk containing $x$ in the partitioning $p$.

An input to the EC problem is an initial partitioning $p_0$ and an online sequence of
sets $R_1, R_2, \dots, R_m$. Time is split into $m$ steps, and when set~$R_t$
arrives in step $t$:
\begin{itemize}
\item \ALG pays an \emph{access cost} $\min_{x \in R_t} \size(p_{t-1}, x)$.
\item \ALG chooses a valid partitioning $p_t: U \to [w]$. 
For each element $x$ with $p_t(x) \neq p_{t-1}(x)$, \ALG pays 
a \emph{movement cost} equal to $\max \{ \size(p_{t-1}, x), \size(p_{t}, x) \}$.
\end{itemize}

\begin{theorem}
\label{thm:ec_to_mssc}
If there exists a $c$-competitive (deterministic or randomized) algorithm for
$\ALGH$ for the EC problem, then there exists an $O(c)$-competitive
(deterministic or randomized) algorithm $\ALGS$ for the MSSC problem. 
\end{theorem}


\subsection{Canonic partitioning}

To prove Theorem~\ref{thm:ec_to_mssc}, note that MSSC and EC problems are
closely related by a natural transformation from any permutation $\pi$ in the
MSSC problem to a partitioning $p$ in the EC problem: Assume that the elements
are ordered in a list according to~$\pi$. Then, $S^{p}_0$ contains the first
list element, $S^{p}_1$ the next $2^1$ elements, $S^{p}_2$ the next $2^2$
elements, and so on, with $S^{p}_{w-1}$ containing the last $2^{w-1}$ elements.

Formally, any permutation $\pi$ of the MSSC induces a~\emph{canonic
partitioning} $\pt(\pi)$ of the EC problem, defined as 
\begin{align*}
    \pt(\pi)(x) = \lfloor \log_2 \pi(x) \rfloor && \text{for any $x \in U$.}
\end{align*}
Note that for any permutation $\pi$ and element $x$ it holds that
\begin{equation}
\label{eq:pi_vs_part}
    \size(\pt(\pi), x) \leq \pi(x) \leq 2 \cdot \size(\pt(\pi), x) - 1.
\end{equation}

\subsection{Constructing online algorithm}

To show Theorem~\ref{thm:ec_to_mssc}, we need to construct 
an algorithm $\ALGS$ for the MSSC problem on the basis of 
an existing algorithm $\ALGH$ for the EC problem.

Observe that any input $\IS = (\pi_0, R_1, R_2, \dots, R_m)$ to the MSSC problem
has a corresponding input $\IH = (p_0 = \pt(\pi_0), R_1, R_2, \dots, R_m)$ to
the EC problem. To provide a~solution to an input~$\IS$, our algorithm \ALGS
internally executes an algorithm \ALGH on input $\IH$. Once \ALGH responds to
$R_t$ by changing its partitioning from $p_{t-1}$ to $p_t$, \ALGS mimics these
changes by modifying its permutation $\pi_{t-1}$ into $\pi_t$, so that
$\pt(\pi_t) = p_t$. 

As we show below, such a definition together with \eqref{eq:pi_vs_part} guarantees
that the \emph{access} costs of \ALGS and \ALGH are equal up to a factor of $2$.
Note that there are multiple ways of obtaining a~permutation $\pi_t$ satisfying
$\pt(\pi_t) = p_t$. The crux is to show that it is possible to choose $\pi_t$,
so that the \emph{reordering} cost of \ALGS is at most constant times higher
than the \emph{movement} cost of \ALGH.


\begin{lemma}
\label{lem:ec_to_mssc_1}
For any step $t$, 
it is possible to choose $\pi_t$, such that $\pt(\pi_t) = p_t$ and 
$\ALGS(\IS, t) \leq 4 \cdot \ALGH(\IH, t)$.
\end{lemma}

\begin{proof}
First, we observe that \ALGS may swap two elements on positions $a \neq b$ using 
$2 \cdot |b-a| - 1$ swaps of adjacent elements, paying 
$2 \cdot |b-a| - 1 < 2 \cdot \max\{a,b\}$.
We use $\swap(a,b)$ to denote both such operation and its cost.

The movements chosen by \ALGH in step $t$ can be expressed by a directed graph,
whose vertices are chunks $S_0, \dots, S_{w-1}$. Each element $x$ that changes
its chunk (from $S_{p_{t-1}(x)}$ to $S_{p_{t}(x)}$) is encoded as a~directed
edge from chunk~$S_{p_{t-1}(x)}$ to $S_{p_{t}(x)}$. As both partitionings
$p_{t-1}$ and $p_{t}$ are valid, the in-degree and out-degree of each vertex are
equal. Thus, the graph can be partitioned into a union of edge-disjoint cycles.
We treat each such cycle separately. For simplicity of the description, we
assume that there is only one such cycle $S_{i_0} \to S_{i_1} \to S_{i_2} \to
\dots \to S_{i_{k-1}} \to S_{i_{k}}$ (where $i_{k} = i_0$). The general case
follows by simply summing over all cycles. 

For any $j \in [k]$, let $x_{j}$ be the element that is moved from 
chunk~$S_{i_j}$ to $S_{i_{j+1}}$; let $v_j = \pi_{t-1}(x_{j})$. To mimic the choices
of~\ALGH, \ALGS executes a sequence of $k-1$ swaps: $\swap(v_{k-1}, v_{k-2})$,
$\swap(v_{k-2}, v_{k-3})$, $\dots$, $\swap(v_2, v_1)$, $\swap(v_1, v_0)$. It is
easy to verify that once these $\swap$ operations are executed, the position of
element $x_j$ becomes equal to $v_{j+1}$ for any $j \in [k-1]$, and the position
of $x_{k-1}$ becomes equal to $v_0$. Thus, the resulting permutation $\pi_t$
satisfies the property $\pt(\pi_t) = p_t$.

To estimate the cost of $\ALGS$, fix any $j \in [k-1]$. By $p_{t-1} =
\pt(\pi_{t-1})$, $p_{t} = \pt(\pi_{t})$, and \eqref{eq:pi_vs_part}, we have 
\begin{align*}
    \swap(v_j, v_{j+1}) 
    & = \swap(\pi_{t-1}(x_j), \pi_{t}(x_j)) \\
    & < 2 \cdot \max \{ \pi_{t-1}(x_j), \pi_{t}(x_j) \} \\
    & < 4 \cdot \max \{ \size(p_{t-1},x_j), \size(p_{t}, x_j) \} \\
    & = 4 \cdot \ALGH(\IH,t,x_j),
\end{align*}
where $\ALGH(\IH,t,x_j)$ is the movement cost of $x_j$ in step~$t$. 

Let $s$ be the element from $R_t$ which is the earliest on the list for permutation
$\pi_{t-1}$. By~\eqref{eq:pi_vs_part} and $p_{t-1} = \pt(\pi_{t-1})$,
we have $\pi_{t-1}(s) < 2 \cdot \size(p_{t-1}, s)$. Summing up, 
\begin{align*}
    \ALGS(\IS,t) 
    & \textstyle = \pi_{t-1}(s) + \sum_{j \in [k-1]} 
        \swap(v_j, v_{j+1}) \\
    & \textstyle < 2 \cdot \size(p_{t-1}, s) + \sum_{j \in [k-1]} 
        4 \cdot \ALGH(\IH,t,x_j)  \\
    & < 4 \cdot \ALGH(\IH,t). 
    \qedhere
\end{align*}
\end{proof}

\subsection{Proof of Theorem 1}

Let $\OPTS$ and $\OPTH$ be the optimal solutions for inputs $\IS$ and $\IH$,
respectively. To show the competitive ratio of $\ALGS$, it remains to relate the
costs of these optimal solutions. 

We say that an algorithm is \emph{move-to-front based (MTF-based)} if, in
response to $R_t$, it chooses exactly one of the elements from $R_t$ and brings
it to the list front; furthermore, it does not perform any further list
reordering.

\begin{lemma}
\label{lem:ec_to_mssc_2}
For any input $\IS$ for the MSSC problem, there exists an (offline) MTF-based 
solution \MTFS, such that $\MTFS(\IS) \leq 2 \cdot \OPTS(\IS)$. 
\end{lemma}

\begin{proof}
Based on the actions of \OPTS on $\IS = (\pi_0, R_1, \dots, R_m)$, 
we may create an input $\JS = (\pi_0, R'_1, \dots, R'_m)$, 
where $R'_i$ is a singleton set containing exactly the element from $R_i$ 
that is closest to the front on the list of $\OPTS$.

Clearly, $\OPTS(\JS) = \OPTS(\IS)$. Furthermore, $\JS$ is an~instance of the
list update problem, for which it is known that moving the requested element to
the list front is a 2-approximation~\cite{SleTar85}. Thus, $\MTFS(\JS) \leq 2
\cdot \OPT(\JS)$. Finally, we observe that reordering actions of $\MTFS(\JS)$
can be also applied to input $\IS$. While the movement cost remains then the
same, the access cost can be only smaller, i.e., $\MTFS(\IS) \leq \MTFS(\JS)$.
The lemma follows by combining the shown inequalities.
\end{proof}

\begin{lemma}
\label{lem:ec_to_mssc_3}
For any input $\IS$ for the MSSC problem, and the associated input $\IH$ for the
EC problem, $\OPTH(\IH) \leq 6 \cdot \MTFS(\IS)$.
\end{lemma}

\begin{proof}
To show the lemma, it suffices to show that there exists an offline algorithm
$\OFFH$ satisfying $\OFFH(\IH, t) \leq 6 \cdot \MTFS(\IS, t)$ for any step $t$.

Let $\OFFH$ be an (offline) algorithm for $\IH$ that in step~$t$ takes the
permutation $\piMTF_t$ of $\MTFS$ and changes its partitioning
to~$\pt(\piMTF_t)$. We now compare the costs of $\OFFH$ to $\MTFS$ in step $t$,
separately for access costs and movement/reordering costs.

Let $s$ be the element from $R_t$ that $\MTFS$ has closest to the list front.
The access cost of \OFFH is $\size(\pt(\piMTF_{t-1}), s)$ which by 
\eqref{eq:pi_vs_part} is at most $\piMTF_{t-1}(s)$, the access cost of \MTFS.

Let $x$ be the element that $\MTFS$ moves to the list front and let $v =
\piMTF_{t-1}(x)$. If $v = 1$, then neither $\MTFS$ nor $\OFFH$ perform any
reordering/movement, and the lemma follows. Thus, we assume that $v \geq 2$. To
move $x$ to the list front, $\MTFS$ executes $v - 1$ swaps; this reordering
increments the positions of all elements that originally preceded $x$. 

To estimate the cost of \OFFH, we analyze the movement costs associated with
changing partitioning from $\pt(\piMTF_{t-1})$ to $\pt(\piMTF_{t})$. Let $\ell =
\lfloor \log_2 v \rfloor \geq 1$. When $x$ is moved to the front, its chunk
changes from $S_\ell$ to $S_0$. To describe the remaining changes, we assume
that the list is ordered from left to right with the list front on the left. Then,
the rightmost element of any chunk $S_0, S_1, \dots, S_{\ell-1}$ changes its
chunk to the next one. Hence, the movement cost of \OFFH is
\begin{align*}
    \textstyle \max & \textstyle \{ 2^\ell, 2^0 \} + \sum_{j \in [\ell]} \max \{ 2^j, 2^{j+1} \} 
    < 3 \cdot 2^\ell \leq 3 \cdot v \leq 6 \cdot (v-1),
\end{align*}
which is $6$ times the reordering cost of \MTFS.
\end{proof}

\begin{proof}[Proof of Theorem~\ref{thm:ec_to_mssc}]
Let $\ALGS$ be defined as in Lemma~\ref{lem:ec_to_mssc_1}. Then,
\begin{align*}
    \ALGS(\IS) 
    & \leq 4 \cdot \ALGH(\IH) 
      \leq 4 \cdot c \cdot \OPTH(\IH) \\ 
    & \leq 24 \cdot c \cdot \MTFS(\IS) 
    \leq 48 \cdot c \cdot \OPTS(\IS).
\end{align*}
The inequalities follow by summing Lemma~\ref{lem:ec_to_mssc_1} over all steps,
$c$-competitiveness of \ALGH, Lemma~\ref{lem:ec_to_mssc_3}, 
and finally by Lemma~\ref{lem:ec_to_mssc_2}.
\end{proof}


\section{Solving Exponential Caching}
\label{sec:algorithm}

In this section, we provide an $O(r^2)$-competitive randomized algorithm for the 
Exponential Caching problem, where $r$ is the maximum cardinality 
of requested sets. By Theorem~\ref{thm:ec_to_mssc}, this will yield an 
$O(r^2)$-competitive algorithm for the Min-Sum Set Cover problem. 
We note that our algorithms do not require prior knowledge about $r$.

In the following description, we skip $t$ subscripts in the notations 
and use $p$ as the \emph{current} value of the partition function,
and $S_i$ as the \emph{current} contents of an appropriate chunk. 

\SetAlgorithmName{Routine}{routine}{List of routines}
\begin{algorithm}[tb]
\caption{\fetch(z), where $z$ is any element}
\label{alg:fetch}
\If{$p(z) > 0$}{
    $\ell \gets p(z)$ \\
    \For{$i = 0, 1, \dots, \ell-1$}{
        $a_i \gets \text{random element of $S_i$}$ \\
    }
    \textbf{move} $z$ from $S_{\ell}$ to $S_0$ \\
    \For{$i = 0, 1, \dots, \ell-1$}{
        \textbf{move} $a_i$ from $S_i$ to $S_{i+1}$ \\
    }
}
$b(z) \gets 0$
\end{algorithm} 

Our algorithm \textsc{Lazy-Move-All-To-Front} (\LM) maintains budget $b(z)$ for
any element $z \in U$. Initially, all budgets are set to zero.

At certain times, \LM wants to move an element $z$ to chunk $S_0$. However,
to preserve the cardinality of $S_0$, it needs to make space in $S_0$. It does
so using a procedure \fetch(z) defined in Routine~\ref{alg:fetch}. This procedure 
chooses a~random sequence of elements and moves them to chunks of larger indexes. 
It also moves $z$ to $S_0$ and resets its budget to zero. 

To serve a set $R = \{x, y_0, y_1, \dots, y_{q-2} \}$, where $q \leq r$ and
$x$ is an element of $R$ with the smallest chunk index, \LM executes $\fetch(x)$
moving $x$ to $S_0$. A natural strategy would be then to move the
remaining elements $y_i$ towards chunks with smaller indexes. However, such
an approach leads to a huge competitive ratio. Instead, \LM performs these movements
in a lazy manner: it increases the budgets of the remaining elements and moves
the elements to $S_0$ once their budgets reach a certain threshold. The details of \LM are given
in Algorithm~\ref{alg:alg}.

\SetAlgorithmName{Algorithm}{algorithm}{List of algorithms}
\begin{algorithm}[tb]
\caption{\textsc{Lazy-Move-All-To-Front}\\
\textbf{Input: } Set $R =  \{x, y_0, y_2, \dots, y_{q-2} \}$, where $q \leq r$ and
$p(x) \leq p(y_i)$ for $i \in [r-1]$} 
\label{alg:alg}
\textbf{pay} access cost $\size(p, x) = 2^{p(x)}$ \label{line:first} \\
\textbf{execute} $\fetch(x)$ \\
\For{$i = 0, 1, \dots, q-2$}{
    $b(y_i) \leftarrow b(y_i) + 2^{p(x)}$ \label{line:budget_increase}
}
\While{exists $z$ such that $b(z) \geq 2^{p(z)}$}{  \label{line:last_line_2}
    \textbf{execute} $\fetch(z)$  \label{line:last_line_3}
}
\end{algorithm}


\subsection{Termination}

We start by showing that the algorithm Algorithm~\ref{alg:alg} is well-defined,
i.e., it terminates. If an~element $z$ satisfies $b(z) \leq 2^{p(z)}$, then we
say that its budget is \emph{controlled}, and it is \emph{uncontrolled}
otherwise.

\begin{observation}
\label{obs:budgets_controlled}
Executing $\fetch(z)$ makes the budget of $z$ controlled and it does not cause
budgets of other elements to become uncontrolled. 
\end{observation}

\begin{proof}
For the elements randomly chosen within routine $\fetch(z)$, their chunk indexes
are increased without changing their budgets which can only make their budgets
controlled. The only element whose chunk index is decreased is $z$ itself, but
its budget is reset to zero, which trivially makes it controlled.

Note that execution of a single $\fetch$ operation may make multiple budgets
controlled.
\end{proof}

By the observation above, the number of elements with uncontrolled budgets
decreases with each iteration of the while loop in Line~\ref{line:last_line_2}
of Algorithm~\ref{alg:alg}. Thus, processing a~set~$R_t$ by algorithm \LM terminates.


\subsection{Potential function}

We compare the cost of \LM to that of an optimal offline solution \OPT. We use
$p$ and $S_i$ to denote the partitioning function and appropriate chunks in the
solution of \ALG, and we use $p^*$ and $S^*_i$ for the corresponding notions in
the solution of~\OPT.

In our analysis, we use four parameters:
$\fa = 7$, $\fc = 7 r - 6$, $\fb = 21r - 11$, $\fk = \lceil \log \fb \rceil$.
Our analysis does not depend on the specific values of these parameters, 
but we require that they satisfy the following relations.

\begin{fact}
\label{fact:constants}
Parameters $\fa$, $\fb$ and $\fc$ satisfy the following relations:
$\fa \geq 7$,
$\fc \geq \fa \cdot (r-2) + 8$,
$\fb \geq \fa \cdot (r-1) + 2 \fc + 8$.
Furthermore, $\fk$ is an integer satisfying $2^\fk \geq \fb$.
\end{fact}

To compute the competitive ratio of \LM, we use amortized analysis. To this end, 
for any element $z \in U$, we define its potential 
\begin{equation*}
    \Phi_z = \begin{cases}
        \fa \cdot b(z) & \text{if $p(z) \leq p^*(x) + \fk$,} \\
        \fb \cdot 2^{p(z)} - \fc \cdot b(z) 
            & \text{if $p(z) \geq p^*(z) + \fk+1$.} 
    \end{cases}
\end{equation*}
We also define the total potential as $\Phi = \sum_{z \in U} \Phi_z$. 


To streamline the proof, we split each step $t$ into two parts. In the first
part (studied in Sections~\ref{sec:fetch} and~\ref{sec:alg_moves}), 
$\LM$ executes Algorithm~\ref{alg:alg} (pays the access cost, chooses
element movements, and pays for them), while \OPT pays the access cost only. In
the second part (studied in Section~\ref{sec:opt_moves}), 
\LM does nothing, while \OPT moves its elements and pays for
these movements. 

We use $\Delta \LM$, $\Delta \OPT$, and $\Delta \Phi$ to denote the increments
in costs of \LM and \OPT and the total potential, respectively, associated with
the currently discussed event. We show that for both step parts, it holds that
$\E[\Delta \LM + \Delta \Phi] \leq O(r^2) \cdot \Delta \OPT$. The competitive
ratio of $O(r^2)$ will then follow by summing this relation for both parts over
all steps of the input.


\subsection{Budget invariant}

We start with a simple bound on the budget values. 

\begin{observation}
\label{obs:budget_invariant}
At any time, for any element $z \in U$, it holds that $b(z) \leq 2 \cdot 2^{p(z)}$.
\end{observation}

\begin{proof}
Note that at the beginning of any step, the budgets of all elements are
controlled ($b(z) \leq 2^{p(z)}$ for any element $z$). This holds trivially at
the beginning as all budgets are then zeros. Moreover, by
Observation~\ref{obs:budgets_controlled}, this property is ensured at the end of
a step by the while loop in Lines \ref{line:last_line_2}--\ref{line:last_line_3}
of~Algorithm~\ref{alg:alg}.

Within a step, the budget of an element may not be controlled because of budget
increases in Line~\ref{line:budget_increase}. However, because of this action,
the budget of $y_i$ may grow only to $2^{p(y_i)} + 2^{p(x)}$, which is at most
$2 \cdot 2^{p(y_i)}$ as $p(x) \leq p(y_i)$.
\end{proof}

Note that $\fb \geq 2 \fc$ by Fact~\ref{fact:constants}.
Thus, Observation~\ref{obs:budget_invariant} and the potential definition
immediately imply the following claim. 

\begin{corollary}
\label{cor:non_negative}
At any time, $\Phi_z \geq 0$ for any element $z$.
\end{corollary}


\subsection{Analysis of operation FETCH}
\label{sec:fetch}

To analyze the amortized cost associated with a single operation 
$\fetch$, we start with calculating the change in the potential due to a movement
of a random element.

\begin{lemma}
\label{lem:elem_upgraded}
Fix a chunk index $i \in [w-1]$. Let $a$ be an element chosen uniformly at random from chunk~$S_i$. 
If $a$ is moved to chunk $S_{i+1}$, then $\E[\Delta \Phi_a] \leq 2^{i+2}$. 
The result holds even conditioned on the current partitioning of \ALG.
\end{lemma}

\begin{proof}
We look at all $2^i$ elements from $S_i$ and their chunk indexes in the solution of \OPT.
Let 
\[
    \tilde{S_i} = \{ z \in S_i \;:\; p^*(z) \leq i-\fk \}.
\]
Note that 
\begin{align*}
    |\tilde{S_i}| 
    & \textstyle = \sum_{j \in [i-\fk+1]} |\{ z \in S_i \;:\; p^*(z) = j \}| \\
    & \textstyle = \sum_{j \in [i-\fk+1]} |S_i \cap S^*_j|
    \leq \sum_{j \in [i-\fk+1]} | S^*_j | \\
    & \textstyle \leq \sum_{j \in [i-\fk+1]} 2^j < 2^{i-\fk+1}.
\end{align*}
Hence, when $a$ is chosen randomly from $S_i$, 
\begin{equation}
\label{eq:p_star}
    \Pr[p^*(a) \leq i - \fk] 
    = \frac{|\tilde{S_i}|}{|S_i|} < 2^{i-\fk+1} / 2^i = 2^{-\fk+1}.
\end{equation}

To upper-bound $\E[\Delta \Phi_a \,|\, p^*(a) \leq i-\fk]$, we consider two
cases. By the potential definition, if $p^*(a) \leq i-\fk-1$, then $\Delta
\Phi_{a} = [\fb \cdot 2^{i+1} - \fc \cdot b(a)] - [\fb \cdot 2^{i} - \fc \cdot
b(a)] = \fb \cdot 2^i$. Otherwise, $p^*(a) = i - \fk$, and then $ \Delta
\Phi_{a} = [ \fb \cdot 2^{i+1} - \fc \cdot b(a) ] - \fa \cdot b(a) \leq \fb
\cdot 2^{i+1}$. Hence, in either case 
\[ 
    \E[\Delta \Phi_a \,|\, p^*(a) \leq i-\fk] \leq \fb \cdot 2^{i+1}. 
\]
Again by the potential definition,
\[
    \E[\Delta \Phi_a \,|\, p^*(a) \geq i-\fk+1] = 
    \fa \cdot b(a) - \fa \cdot b(a) = 0.
\]
Combining these bounds on $\Delta \Phi_a$ with \eqref{eq:p_star} yields
\begin{align*}
  \E[\Delta\Phi_{a}] 
  & = \E[\Delta \Phi_{a} \,|\, p^*(a) \leq i - \fk] \cdot \Pr[p^*(a) \leq i - \fk] \\
    & \quad + \E[\Delta \Phi_{a} \,|\, p^*(a) \geq i - \fk + 1] 
        \cdot \Pr[p^*(a) \geq i - \fk + 1] \\
  & < \beta \cdot 2^{i+1} \cdot 2^{-\fk+1} .
\end{align*}
The lemma follows as $\fb \leq 2^\fk$ by Fact~\ref{fact:constants}.
\end{proof}

\begin{lemma}
\label{lem:fetch_cost}
Whenever \LM executes operation $\fetch(z)$, 
$\E[\Delta \LM + \Delta \Phi] \leq 7 \cdot 2^{p(z)} - g$,
where $g$ is the value of $\Phi_z$ right before this operation.
\end{lemma}

\begin{proof}
First, we estimate $\Delta \LM$ itself due to $\fetch(z)$. Recall that the
procedure $\fetch$ creates a~sequence of random elements $a_0, a_1, \dots,
a_{p(z)-1}$, where $a_i \in S_i$ and moves each $a_i$ from chunk $S_i$ to
$S_{i+1}$. Furthermore, $z$ is moved from chunk $S_{p(z)}$ to $S_0$. Thus, the
associated cost is 
\begin{equation}
\label{eq:fetch_alg}
\begin{split}
\Delta \LM 
    & \textstyle = \max\{2^{p(z)}, 2^0\} + \sum_{i=0}^{p(z)-1} \max\{2^i, 2^{i+1}\} \\
    & \textstyle = 2^{p(z)} + \sum_{i=1}^{p(z)} 2^i
    < 3 \cdot 2^{p(z)}. 
\end{split}
\end{equation}

It remains to analyze the potential change for the moved elements: $a_0, a_1,
\dots, a_{p(z)-1}$, and $z$. The potential of $z$ before the movement is equal to $g$
by the lemma assumption. By the definition, the potential of $z$ after the 
movement is $\fa \cdot b(z)$, which is equal to $0$ as the budget of $z$
is reset to $0$ within $\fetch(z)$ operation. Thus, 
\begin{equation}
\label{eq:fetch_delta_phi_z}
    \Delta \Phi_z = -g.
\end{equation}

Finally, by Lemma~\ref{lem:elem_upgraded},
$\E[\Delta \Phi_{a_i}] \leq 4 \cdot 2^i$ for any $i \in [p(z)]$.
Combining that with \eqref{eq:fetch_alg} and \eqref{eq:fetch_delta_phi_z}, and using linearity
of expectation yields
\begin{align*}
    \E[\Delta \LM + \Delta \Phi] 
    & \textstyle = \Delta \LM + \E[\Delta \Phi_z] + \sum_{i=0}^{p(z)-1} \E[\Delta\Phi_{a_i}] \\
    & \textstyle < 3 \cdot 2^{p(z)} - g + \sum_{i=0}^{p(z)-1} 4 \cdot 2^i \\
    & \textstyle < 7 \cdot 2^{p(z)} - g.
\hfill \qedhere
\end{align*}
\end{proof}


\subsection{Amortized cost of LMA}
\label{sec:alg_moves}

Now we may upper bound the total amortized cost of \LM in a single step. 
We split this cost into parts incurred by Lines \ref{line:first}--\ref{line:budget_increase}
and Lines~\ref{line:last_line_2}--\ref{line:last_line_3}.

\begin{lemma}
\label{lem:amortized_cost_56}
Whenever \LM executes Lines~\ref{line:last_line_2}--\ref{line:last_line_3}
of Algorithm~\ref{alg:alg}, it holds that $\E[\Delta \LM + \Delta \Phi] \leq 0$.
\end{lemma}

\begin{proof}
Let $z$ be the element moved in Line~\ref{line:last_line_3}. 
Line~\ref{line:last_line_2} ensures 
that $b(z) \geq 2^{p(z)}$. 
Furthermore, $b(z) \leq 2 \cdot 2^{p(z)}$ by Observation~\ref{obs:budget_invariant}.
Let $\Phi_z$ be the value of the potential right before operation $\fetch(z)$ is 
executed in Line~\ref{line:last_line_3}. 
By the potential definition,
\begin{align*}
    \Phi_z 
    & \geq \min \{ \fb \cdot 2^{p(z)} - \fc \cdot b(z) \,, 
    \fa \cdot b(z) \} \\
    & \geq \min \{ \fb - 2 \cdot \fc  \,,
    \fa \} \cdot 2^{p(z)}\\
    & \geq 7 \cdot 2^{p(z)} && \text{(by Fact~\ref{fact:constants})}.
\end{align*}
By Lemma~\ref{lem:fetch_cost},
$\E[\Delta \LM + \Delta \Phi] \leq 7 \cdot 2^{p(z)} - \Phi_z \leq 0$.
\end{proof}

\begin{lemma}
\label{lem:amortized_cost}
Fix any step and consider its first part, where \LM 
pays for its access and movement costs, whereas \OPT pays for 
its access cost. Then, 
$\E[\Delta \LM + \Delta \Phi] \leq 
(\fa \cdot (r-1) + 8) \cdot 2^\fk \cdot \Delta \OPT = 
O(r^2) \cdot \Delta \OPT$.
\end{lemma}

\begin{proof}
Let $R = \{x, y_0, \dots, y_{q-2} \}$ be the requested set, where $q \leq r$ 
and $p(x) \leq p(y_i)$ for any $i \in [q-1]$. 
Let $\Phi_x, \Phi_{y_0}, \dots, \Phi_{y_{q-2}}$ be the potentials of elements from $R$
just before the request.

It suffices to analyze the amortized cost of \LM in Lines~\ref{line:first}--\ref{line:budget_increase},
as the cost in the subsequent lines is at most~$0$ by Lemma~\ref{lem:amortized_cost_56}.
The access cost paid by \LM is $2^{p(x)}$ and by Lemma~\ref{lem:fetch_cost}, 
the amortized cost of $\fetch(x)$ is
$7 \cdot 2^{p(x)} - \Phi_x$.
Therefore,
\begin{equation}
\label{eq:alg_cost_1}
    \textstyle \E[\ALG + \Delta \Phi]
    = 8 \cdot 2^{p(x)} - \Phi_x + \sum_{i \in [q-1]} \Delta \Phi_{y_i} .
\end{equation}
As $b(y_i)$ grows by $2^{p(x)}$ for any $i \in [q-1]$, 
\begin{align}
    \label{eq:alg_cost_2}
    \Delta \Phi_{y_i} \leq \fa \cdot 2^{p(x)} && \text{for any $i \in [q-1]$}.
\end{align}
Finally, by Corollary~\ref{cor:non_negative},
\begin{equation}
    \label{eq:alg_cost_3}
    \Phi_x \geq 0.
\end{equation}

Let $w \in R$ be the element with the smallest chunk index in the solution of \OPT. 
That is, $\Delta \OPT = 2^{p^*(w)}$.

Assume first that $p(x) \leq p^*(w) + \fk$. By \eqref{eq:alg_cost_1},
\eqref{eq:alg_cost_2}, and~\eqref{eq:alg_cost_3}, $\E[\Delta \LM + \Delta \Phi]
\leq (8 + \fa \cdot (q-1)) \cdot 2^{p(x)} \leq (8 + \fa \cdot (r-1)) \cdot 2^\fk
\cdot \Delta \OPT$, and thus the lemma follows. 

Therefore, in the remaining part of the proof, we assume that 
$p(x) \geq p^*(w)+\fk+1$ and we show that, in such case, 
$\E[\Delta \LM + \Delta \Phi] \leq 0$.
We consider two cases. 
\begin{itemize}
\item If $w = x$, we may use a stronger lower bound on $\Phi_x$, i.e.,
    $\Phi_x = \fb \cdot 2^{p(x)} - \fc \cdot b(x) \geq (\fb - 2 \fc) \cdot 2^{p(x)}$
    (cf.~Observation~\ref{obs:budget_invariant}).
    Together with \eqref{eq:alg_cost_1} and
    \eqref{eq:alg_cost_2}, this yields
    $\E[\Delta \LM + \Delta \Phi] \leq (8 + \fa \cdot (q-1) - \fb + 2 \fc) 
    \cdot 2^{p(x)}$.
\item If $w = y_j$ for some $j \in [q-1]$, then 
    as $p(y_j) \geq p(x) \geq p^*(y_j)+\fk+1$, we
    may use a stronger upper bound on $\Delta \Phi_{y_j}$, namely
    $\Delta \Phi_{y_j} \leq - \fc \cdot 2^{p(x)}$.
    Together with 
    \eqref{eq:alg_cost_1}, \eqref{eq:alg_cost_2} (for $i \neq j$) and~\eqref{eq:alg_cost_3},
    this yields 
    $\E[\Delta \LM + \Delta \Phi] \leq (8 +\fa \cdot (q-2) - \fc) \cdot 2^{p(x)}$.
\end{itemize}
In either case, Fact~\ref{fact:constants} together with $q \leq r$ ensures
that $\E[\Delta \LM + \Delta \Phi] \leq 0$.
\end{proof}


\subsection{Movement of OPT}
\label{sec:opt_moves}

\begin{lemma}
\label{lem:delta_phi_when_opt_moves}
Fix any step and consider its second part, where \LM does nothing, 
whereas \OPT moves elements and pays for their movement. 
Then $\Delta \LM + \Delta \Phi \leq (2 \fa + 2 \fc + \fb) \cdot 2^\fk \cdot \Delta \OPT 
= O(r^2) \cdot \Delta \OPT$.
\end{lemma}

\begin{proof}
We focus on a single element $z$ moved by \OPT. 
Assume that \OPT changes its chunk
index~$p^*(z)$ from $a$ to $a+d$ (where $d$ is possibly negative).

The only element whose potential might be affected is $z$ itself. 
The definition of $\Phi_z$ has two cases, depending on whether 
the relation $p(z) \leq p^*(z) + k$ holds. If this relation remains untouched by the movement, 
then $\Phi_z$ remains constant. On the other hand, the relation changes
only in one of the following cases.
\begin{itemize}
\item $d$ is positive, $a \leq p(z) - \fk-1$, and $a + d \geq p(z) - \fk$;
\item $d$ is negative, $a \geq p(z) - \fk$, and $a + d \leq p(z) - \fk - 1$.
\end{itemize}
In the first case, $p(z) \leq \fk + a + d$, while in the second case 
$p(z) \leq \fk + a$. Thus, in either case, $p(z) \leq \fk + \max\{a, a + d\}$.
We obtain
\begin{align*}
    \Delta \Phi_z 
    & \leq | \fa \cdot b(z) - \fb \cdot 2^{p(z)} + \fc \cdot b(z) | \\
    & \leq (2 \fa + 2 \fc + \fb) \cdot 2^{p(z)} \\
    & \leq (2 \fa + 2 \fc + \fb) \cdot 2^\fk \cdot \max\{2^a, 2^{a+d} \} ,
\end{align*}
where the second inequality is implied by Observation~\ref{obs:budget_invariant}.
Note that the cost of \OPT associated with moving $z$ is $\max \{2^a, 2^{a+d} \}$
by the definition of the EC problem. 
Summing over all elements moved by \OPT immediately yields 
$\Delta \Phi \leq (2 \fa + 2 \fc + \fb) \cdot 2^\fk \cdot \Delta \OPT$.
The lemma follows as $\Delta \LM = 0$ in the second part of a step.
\end{proof}


\subsection{Competitiveness}

\begin{theorem}
$\LM$ is $O(r^2)$-competitive for the Exponential Caching problem, even against 
adaptive-online adversaries.
\end{theorem}

\begin{proof}
Fix any input $\I$ and consider any step $t$. 
Let $\Phi^t$ denote the potential right after step $t$, and $\Phi^0$ 
be the initial potential. 
By Lemmas~\ref{lem:amortized_cost} and~\ref{lem:delta_phi_when_opt_moves},
\begin{align}
\label{eq:single_step}
    \E[\LM(\I,t) & + \Phi^t - \Phi^{t-1}] = O(r^2) \cdot \OPT(\I,t).
\end{align}

By summing \eqref{eq:single_step} over all $m$ steps of the input, we obtain
that $\E[\LM(\I)] + \E[\Phi^m] - \Phi^0 \leq O(r^2) \cdot \OPT(\I)$. As the
initial potentials of all elements are $0$ and the final potentials are
non-negative by Corollary~\ref{cor:non_negative}, $\E[\LM(\I)] \leq O(r^2) \cdot
\OPT(\I)$. 

We note that the only place where \LM uses randomness is in choosing a sequence of random 
elements in the procedure $\fetch$. As noted in its analysis (cf.~Lemma~\ref{lem:elem_upgraded}),
the bound on the expected amortized cost of \LM holds also conditioned on its current state,
and thus it holds even if the adversary chooses the requested set $R_t$ on the basis of the random 
bits of \LM used till step~$t-1$. 
\end{proof}

The result of Ben-David et al.~\shortcite{BeBKTW94} shows that the existence of a randomized algorithm
that is $c$-competitive against adaptive-online adversaries implies the existence 
of a $c^2$-competitive deterministic algorithm. 

\begin{corollary}
There exists an $O(r^4)$-competitive deterministic algorithm for the Exponential Caching problem.
\end{corollary}

Finally combining the results for Exponential Caching with Theorem~\ref{thm:ec_to_mssc}
immediately gives improved guarantees for the online Min-Sum Set Cover problem. 

\begin{theorem}
There exist a randomized $O(r^2)$-competitive algorithm and a deterministic
$O(r^4)$-competitive algorithm for the Min-Sum Set Cover problem.
\end{theorem}


\section{Conclusions}

In this paper, we studied the online Min-Sum Set Cover problem on a universe of
$n$ elements with requested sets of cardinalities at most $r$. We gave a first
(randomized) algorithm whose competitive ratio does not depend on $n$: our
algorithm is $O(r^2)$-competitive. While our construction implies also the
existence of $O(r^4)$-competitive deterministic solution, it is unknown how to
make it constructive and efficient.

Closing the gaps between our results and lower bounds is an intriguing open question:
while the deterministic lower bound is $\Omega(r)$, no super-constant 
randomized bounds are known.


\section*{Acknowledgements}

This paper was supported by Polish National Science Centre 
grant~2016/22/E/ST6/00499 and by the ERC CoG grant TUgbOAT no 772346.

\bibliography{ref}
\end{document}